%% file: main.tex
\def\year{2019}\relax
\documentclass[letterpaper]{article} 

\usepackage{aaai19}

\usepackage{times}
\usepackage{xcolor}
\usepackage{soul}
\usepackage[utf8]{inputenc}
\usepackage[small]{caption}
\usepackage{balance}

\usepackage{url}  
\usepackage{graphicx}  
\usepackage{mathtools}
\usepackage{amsthm}
\usepackage{amssymb} 
\usepackage{algorithm}
\usepackage[noend]{algorithmic}
\usepackage{xspace}
\usepackage{pgf}
\usepackage{tikz}
\usetikzlibrary{arrows,automata}

\usepackage{float}
\restylefloat{table}
\usepackage{times} 
\usepackage{enumitem} 
\input {macros.tex}

\begin{document}
\title{SAT-based Explicit $\ltlf$ Satisfiability Checking\thanks{
Geguang Pu and Kristin Y. Rozier are corresponding authors.}
}


\author{Jianwen Li, Kristin Y. Rozier\\ Iowa State University\\ Ames, IA, USA\\\{jianwen,kyrozier\}@iastate.edu 
\And Geguang Pu, Yueling Zhang\\ East China Normal University\\ Shanghai, China\\\{ggpu,ylzhang\}@sei.ecnu.edu.cn
\And Moshe Y. Vardi\\ Rice University\\ Houston, TX, USA\\ vardi@cs.rice.edu}

\maketitle

\vspace{-0.5in}
\begin{abstract}
We present here a SAT-based framework for $\ltlf$ (Linear Temporal Logic on Finite Traces) satisfiability checking. We use propositional SAT-solving techniques to construct a transition system for the input $\ltlf$ formula; satisfiability checking is then reduced to a path-search problem over this transition system. Furthermore, we introduce \cdlsc (Conflict-Driven $\ltlf$ Satisfiability Checking), a novel algorithm that leverages information produced by propositional SAT solvers from both satisfiability and unsatisfiability results. Experimental evaluations show that \cdlsc outperforms all other existing approaches for $\ltlf$ satisfiability checking, by demonstrating an approximate four-fold speed-up compared to the second-best solver. 
\end{abstract}

\input {introduction.tex}

\input {preliminary.tex}

\input {overview.tex}

\input {explicit.tex}

\input {cdlsc.tex}

\input {experiment.tex}

\input {conclusion.tex}

\noindent{\bf Acknowledgement.} 
We thank anonymous reviewers for the helpful comments. 
This work is supported in part by NASA ECF NNX16AR57G, NSF CAREER Award CNS-1552934, 
NSF grants CCF-1319459 and IIS-1527668, 
NSF Expeditions in Computing project ``ExCAPE: Expeditions in 
Computer Augmented Program Engineering'', NSFC projects No. 6157297, No. 61632005 and No. 61532019, 
and China HGJ project No. 2017ZX01038102-002.

\clearpage
\pagebreak
\balance
\bibliographystyle{named}
\small 
\inputencoding{latin2}
\bibliography{ok,cav,symbolic}
\inputencoding{utf8}

\input{append.tex}

\end{document}

%% file: macros.tex
 \newcommand{\F}{\mathcal{F}}
\newcommand{\G}{\mathcal{G}}

 \newcommand{\N}{\mathcal{N}}
 
 \newcommand{\R}{\mathcal{R}}
 
\newcommand{\U}{\mathcal{U}} 
 \newcommand{\X}{\mathcal{X}}

\newcommand{\Nat}{{\rm I\kern-.23em N}}

\newcommand{\ltlf}{\mathsf{LTL}_f}
\newcommand{\ff}{\mathsf{ff}}
\renewcommand{\tt}{\mathsf{tt}}
\newtheorem{theorem}{Theorem}
\newtheorem{lemma}{Lemma}

\newtheorem{definition}{Definition}

\newcommand{\xnf}[1]{\textsf{xnf}(#1)}
\newcommand{\tnf}[1]{\textsf{tnf}(#1)}
\newcommand{\tran}[1]{\xrightarrow[]{#1}}

\newcommand{\cdlsc}{\textsf{CDLSC}\xspace}

\newcommand{\cs}{\mathcal{C}}
\newcommand{\p}[1]{\textsf{PA}(#1)}

%% file: introduction.tex
\section{Introduction}\label{sec:intro}

Linear Temporal Logic over Finite Traces, or $\ltlf$, is a formal language gaining popularity in the AI community for formalizing and validating system behaviors. While standard Linear Temporal Logic (LTL) is interpreted on infinite traces \cite{Pnu77}, $\ltlf$ is interpreted over finite traces \cite{GV13}.  While LTL is typically used in formal-verification settings, where we are interested in nonterminating computations, cf. \cite{Var07a},  $\ltlf$ is more attractive in AI scenarios focusing on finite behaviors, such as planning \cite{BK98,DV99,CDV02,PLGG11,CBMM17}, plan constraints \cite{BK00,Gab04}, and user preferences \cite{BFM06,BFM11,SBM11}. Due to the wide spectrum of applications of $\ltlf$ in the AI community \cite{DMM14}, it is worthwhile to study and develop an efficient framework for solving $\ltlf$-reasoning problems. Just as propositional satisfiability checking is one of the most fundamental propositional reasoning tasks, $\ltlf$ satisfiability checking is a fundamental task for $\ltlf$ reasoning. 

Given an $\ltlf$ formula, the satisfiability problem asks whether there is a finite trace that satisfies the formula. A ``classical'' solution to this problem is to reduce it to the LTL satisfiability problem \cite{GV13}. The advantage of this approach is that the LTL satisfiability problem has been studied for at least a decade, and many mature tools are available, cf.~\cite{RV07,RV12}. Thus, $\ltlf$ satisfiability checking can benefit from progress in LTL satisfiability checking. There is, however, an inherent drawback that an extra cost has to be paid when checking LTL formulas, as the tool searches for a ``lasso'' (a lasso consists of a finite path plus a cycle, representing an infinite trace), whereas models of $\ltlf$ formulas are just finite traces. Based on this motivation, \cite{LZPVH14} presented a tableau-style algorithm for $\ltlf$ satisfiability checking. They showed that the dedicated tool, \emph{Aalta-finite}, which conducts an explicit-state search for a satisfying trace, outperforms extant tools for $\ltlf$ satisfiability checking. 

The conclusion of a dedicated solver being superior to $\ltlf$ satisfiability checking from~\cite{LZPVH14}, seems to be out of date by now because of the recent dramatic improvement in propositional SAT solving, cf.~\cite{MZ09}. On one hand, SAT-based techniques have led to a significant improvement on LTL satisfiability checking, outperforming the tableau-based techniques of \emph{Aalta-finite} \cite{LZPVH14}. (Also, the SAT-based tool \emph{ltl2sat} for $\ltlf$ satisfiability checking outperforms \emph{Aalta-finite} on particular benchmarks \cite{FG16}.) On the other hand, SAT-based techniques are now dominant in symbolic model checking \cite{CCDGMMMRT14,VWM15}. Our preliminary evaluation indicates that $\ltlf$ satisfiability checking via SAT-based model checking  \cite{Bra11,EMB11} or via SAT-based LTL satisfiability checking \cite{LZPV15} both outperform the tableau-based tool \emph{Aalta-finite}.
Thus, the question raised initially in \cite{RV07} needs to be re-opened with respect to $\ltlf$ satisfiability checking: is it best to reduce to SAT-based model checking or develop a dedicated SAT-based tool?

Inspired by \cite{LZPV15}, we present an explicit-state SAT-based framework for $\ltlf$ satisfiability. We construct the \emph{$\ltlf$ transition system} by utilizing SAT solvers to compute the states explicitly. Furthermore, by making use of both satisfiability and unsatisfiability information from SAT solvers, we propose a \emph{conflict-driven} algorithm, \cdlsc, for efficient $\ltlf$ satisfiability checking. We show that by specializing the transition-system approach of \cite{LZPV15} to $\ltlf$ and its finite-trace semantics, we get a framework that is significantly simpler and yields a much more efficient algorithm \cdlsc than the one in \cite{LZPV15}.

We conduct a comprehensive comparison among different approaches. 
Our experimental results show that the performance of \cdlsc dominates all other existing $\ltlf$-satisfiability-checking algorithms. On average, \cdlsc achieves an approximate four-fold speed-up, compared to the second-best solution (IC3 \cite{Bra11}+K-LIVE \cite{CS12}) tested in our experiments. Our results re-affirm the conclusion of \cite{LZPVH14} that the best approach to $\ltlf$ satisfiability solving is via a dedicated tool, based on explicit-state techniques.



%% file: preliminary.tex

\section{LTL over Finite Traces}\label{sec:pre}
Given a set $\mathcal{P}$ of atomic propositions, an $\ltlf$ formula
$\phi$ has the form:

$\ \ \ \ \ \ \ \ \ \ \ \phi ::= \tt\ |\ p\ |\ \neg \phi\ |\ \phi\wedge\phi\ |\ \X\phi\ |\ \phi \U\phi$;

where $\tt$ is true, $\neg$ is the negation operator, $\wedge$ is the and operator, $\X$ is the strong Next operator 
and  $\U$ is the Until operator. We also have the duals $\ff$ (false) for $\tt$, $\vee$ for $\wedge$, 
$\N$ (weak Next) for $\X$ and $\R$ for $\U$. A \emph{literal} is an atom $p\in\mathcal{P}$ or its negation ($\neg p$).
 Moreover, we use the notation $\G\phi$ (Globally) and $\F\phi$
(Eventually) to represent $\ff \R\phi$ and $\tt \U\phi$.
Notably, $\X$ is the standard \emph{next} operator, while $\N$ is \emph{weak next}; 
$\X$ requires the existence of a successor state, while $\N$ does not. 
Thus $\N\phi$ is always true in the last state of a finite trace, since no successor exists there.
This distinction is specific to $\ltlf$.

$\ltlf$ formulas are interpreted over finite traces \cite{GV13}.
Given an atom set $\mathcal{P}$, we define $\Sigma = 2^{\mathcal{P}}$ be  the family of sets of atoms. Let
$\xi\in\Sigma^+$ be a finite nonempty trace, with $\xi=\sigma_0\sigma_1\ldots\sigma_n$. we use
$|\xi|=n+1$ to denote the length of $\xi$. Moreover, for $0\leq
i\leq n$, we denote $\xi[i]$ as the i-th position of $\xi$, and 
$\xi_i$ to represent $\sigma_i\sigma_{i+1}\ldots\sigma_n$, which is
the suffix of $\xi$ from position~$i$. We define the satisfaction relation $\xi\models\phi$ as follows:

\begin{itemize}[noitemsep,topsep=0pt]
  \item $\xi\models\tt$; and $\xi\models p$, if $p\in\mathcal{P}$ and $p\in\xi[0]$;
  \item $\xi\models\neg\phi$, if $\xi\not\models\phi$;
  \item $\xi\models\phi_1\wedge\phi_2$,  if $\xi\models\phi_1$ and $\xi\models\phi_2$;
  \item $\xi\models \X\phi$ if $|\xi|>1$ and $\xi_1\models\psi$;
  \item $\xi\models (\phi_1 \U\phi_2)$, if there exists $0\leq i < |\xi|$
  such that $\xi_i\models\phi_2$ and for every $0\leq j < i$ it holds that $\xi_j\models\phi_1$;
  \end{itemize}

\begin{definition}[$\ltlf$ Satisfiability Problem]
Given an $\ltlf$ formula $\phi$ over the alphabet $\Sigma$,
we say $\phi$ is satisfiable iff there is a finite nonempty trace
$\xi\in\Sigma^+$ such that $\xi\models\phi$.
\end{definition}

\noindent\textbf{Notations.}
We use $cl(\phi)$ to denote the set of subformulas of $\phi$. 
Let $A$ be a set of $\ltlf$ formulas, we denote $\bigwedge A$ to be the formula $\bigwedge_{\psi\in A}\psi$.
The two $\ltlf$ formulas $\phi_1,\phi_2$ are semantically equivalent, denoted as $\phi_1\equiv\phi_2 $, iff for every finite trace $\xi$, $\xi\models\phi_1$ iff $\xi\models\phi_2$.  Obviously, we have $(\phi_1\vee\phi_2)\equiv \neg (\neg\phi_1 \wedge\neg \phi_2)$, 
$\N\psi\equiv \neg \X\neg\psi$ and $(\phi_1 \R\phi_2)\equiv \neg (\neg\phi_1 \U\neg \phi_2)$.

We say an $\ltlf$ formula $\phi$ is in \emph{Tail Normal Form} (TNF) if $\phi$ is in \emph{Negated Normal Form} (NNF) and $\N$-free. 
It is trivial to know that every $\ltlf$ formula 
has an equivalent NNF. Assume $\phi$ is in NNF, $\tnf{\phi}$ is defined as $t(\phi)\wedge \F Tail$, where $Tail$ 
is a new atom to identify the last state of satisfying traces (Motivated from \cite{GV13}), and $t(\phi)$ is an $\ltlf$ formula defined recursively as follows: (1) $t(\phi) = \phi$ if $\phi$ is $\tt,\ff$ or a literal; 
(2) $t(\X\psi) = \neg Tail \wedge \X (t(\psi))$; (3) $t(\N\psi) = Tail\vee \X(t(\psi))$; (4) $t(\phi_1\wedge\phi_2) = t(\phi_1)\wedge t(\phi_2)$; (5) $t(\phi_1\vee\phi_2) = t(\phi_1)\vee t(\phi_2)$; (6) $t(\phi_1 \U\phi_2) = (\neg Tail\wedge t(\phi_1)) \U t(\phi_2)$; (7) $t(\phi_1 \R\phi_2) = (Tail\vee t(\phi_1)) \R t(\phi_2)$. 


\begin{theorem}\label{thm:tnf}
    $\phi$ is satisfiable iff $\tnf{\phi}$ is satisfiable.
\end{theorem}

In the rest of the paper, unless clearly specified, the input $\ltlf$ formula is in TNF.

%% file: overview.tex
\section{Approach Overview}
There is a Non-deterministic 
Finite Automaton (NFA) $\mathcal{A}_{\phi}$ that accepts exactly the same language as an $\ltlf$ formula $\phi$ \cite{GV13}.
Instead of constructing the NFA for $\phi$, we generate  
the corresponding \emph{transition system} (Definition \ref{def:ts}), by leveraging SAT solvers.  
The transition system can be considered as an intermediate structure of the NFA, in which every state consists of 
a set of subformulas of $\phi$. 

The classic approach to generate the NFA from an $\ltlf$ formula, i.e. Tableau Construction \cite{GPVW95}, creates the set of 
all one-transition next states of the current state. However, the number of these states is extremely 
large. To mitigate the overload, we leverage SAT solvers to compute the next states of the current state iteratively. 
Although both approaches share the same 
worst case (computing all states in the state space), our new approach is better for on-the-fly checking, as it 
computes new states only if the satisfiability of the formula cannot be determined based on existing states. 

We show the SAT-based approach via an example. Consider the formula 
$\phi=(\neg Tail \wedge a) \U b$. The initial state $s_0$ of the transition system 
is $\{\phi\}$. To compute the next states of $s_0$, we translate $\phi$ to its equivalent \emph{neXt Normal Form} (XNF),  
e.g. $\xnf{\phi} = (b\vee ((\neg Tail\wedge a)\wedge \X\phi))$, see Definition \ref{def:xnf}. If we replace $\X\phi$ in $\xnf{\phi}$ 
with a new propositions $p_1$, the new formula, denoted $\xnf{\phi}^p$, is a pure Boolean formula. As a result, 
a SAT solver can compute an assignment for the formula $\xnf{\phi}^p$. Assume the assignment is 
$\{a, \neg b, \neg Tail, p_1\}$, then we can induce that 
$(a\wedge \neg b\wedge \neg Tail\wedge \X\phi)\Rightarrow \phi$ is true, which  
indicates $\{\phi\}=s_0$ is a one-transition next state of $s_0$, i.e. $s_0$ has a self-loop with the label $\{a,\neg b, \neg Tail\}$. 
To compute another next state of $s_0$, we add the constraint 
$\neg p_1$ to the input of the SAT solver. Repeat the above process and we can construct all states 
in the transition system. 

Checking the satisfiability of $\phi$ is then reduced to finding a \emph{final state} (Definition \ref{def:final}) 
in the corresponding transition system. 
Since $\phi$ is in TNF, a final state $s$ meets the constraint that 
$Tail\wedge \xnf{\bigwedge s}^p$ (recall $s$ is a set of subformulas of $\phi$) 
is satisfiable. For the above example, the initial state $s_0$ is actually a final state, as $Tail\wedge \xnf{\phi}^p$ is satisfiable. Because all states computed by the SAT solver in the transition system are reachable from the initial 
state, we can prove that $\phi$ is satisfiable iff there is a final state in the system (Theorem \ref{thm:reasoning}). 

We present a conflict-driven algorithm, i.e. \cdlsc, to accelerate the satisfiability checking. 
\cdlsc maintains a \emph{conflict sequence} $\mathcal{C}$, in which each element, denoted as $\mathcal{C}[i]$ 
($0\leq i <|\mathcal{C}|$), is a set of states in the transition system that cannot reach a final state in $i$ steps. 
Starting from the initial state, \cdlsc iteratively checks whether a final state can be reached, and makes use of the 
conflict sequence to accelerate the search. Consider the formula 
$\phi = (\neg Tail) \U a\wedge (\neg Tail) \U (\neg a)\wedge (\neg Tail) \U b\wedge (\neg Tail) \U (\neg b) \wedge (\neg Tail) \U c$. 
In the first iteration, \cdlsc checks whether the initial state $s_0=\{\phi\}$ is a final state, i.e. whether 
$Tail\wedge \xnf{\phi}^p$ is satisfiable. The answer is negative, so $s_0$ cannot reach a final state in 0 steps and  
can be added into $\mathcal{C}[0]$. However, we can do 
better by leveraging the Unsatisfiable Core (UC) returned from the SAT solver. 
Assume that we get the UC $u_1=\{(\neg Tail) \U a, (\neg Tail) \U (\neg a)\}$. That indicates every state $s$ containing $u$, 
i.e. $s\supseteq u$, is not a final state. As a result, we can add $u$ instead of $s_0$ into $\mathcal{C}[0]$, making the algorithm much 
more efficient. 

Now in the second iteration, \cdlsc first tries to compute a one-transition next state of $s_0$ that is not included in $\mathcal{C}[0]$. 
(Otherwise the new state 
cannot reach a final state in 0 step.) This can be encoded as a Boolean formula $\xnf{\phi}^p\wedge\neg (p_1\wedge p_2)$ where $p_1, p_2$ represent $\X((\neg Tail) \U a)$ and $\X((\neg Tail) \U (\neg a))$ respectively. Assume the new state $s_1=\{(\neg Tail) \U a, (\neg Tail) \U b, (\neg Tail) \U (\neg b), (\neg Tail) \U c\}$ is generated from the assignment of 
the SAT solver. Then \cdlsc checks whether $s_1$ can reach a final state in 0 step, i.e. $\xnf{\bigwedge s_1}^p\wedge Tail$ is satisfiable. 
The answer is negative and we can add the UC $u_2=\{(\neg Tail) \U b, (\neg Tail) \U (\neg b)\}$ to $\mathcal{C}[0]$ as well. 
Now to compute a next state of $s_0$ that is not included in $\mathcal{C}[0]$, the encoded Boolean formula becomes 
$\xnf{\phi}^p\wedge\neg (p_1\wedge p_2)\wedge \neg (p_3\wedge p_4)$ where $p_3$, $p_4$ represent $\X((\neg Tail) \U b)$ 
and $\X((\neg Tail) \U (\neg b))$ respectively. Assume the new state $s_2=\{(\neg Tail) \U a, (\neg Tail) \U b, (\neg Tail) \U c\}$ is generated from the assignment of the SAT solver. Since $\xnf{\bigwedge s_2}^p\wedge Tail$ is satisfiable, $s_2$ is a final state and we conclude that the formula $\phi$ is satisfiable. In principle, there are a total of $2^5=32$ states in the transition system of $\phi$, 
but \cdlsc succeeds to find the answer by computing only 3 of them (including the initial state).  

\cdlsc also leverages the conflict sequence to accelerate checking unsatisfiable formulas. 
Similar to Bounded Model Checking (BMC) \cite{BCCZ99}, \cdlsc searches the model iteratively. However, BMC invokes only 1 SAT call 
for each iteration, while \cdlsc invokes multiple SAT calls. \cdlsc is more like an IC3-style algorithm, but achieves a 
much simpler implementation by using UC instead of the \emph{Minimal Inductive Core} (MIC)  like IC3 \cite{Bra11}.

%% file: explicit.tex
\section {SAT-based Explicit-State Checking}\label{sec:explicit}
Given an $\ltlf$ formula $\phi$, we construct the \emph{$\ltlf$ transition system} \cite{LZPVH14,LZPV15} by SAT solvers and then check the satisfiability of the formula over its corresponding transition system. 

\subsection{$\ltlf$ Transition System}

First, we show how one can consider $\ltlf$ formulas as propositional ones. This requires considering temporal subformulas as \textit{propositional atoms}.  
 
\begin{definition}[Propositional Atoms]\label{def:propatoms}
For an $\ltlf$ formula $\phi$, we define the set of \textit{propositional atoms} of $\phi$, 
i.e. $\p{\phi}$, as follows: (1) $\p{\phi}=\{\phi\}$ if $\phi$ is an atom, Next, Until or Release formula; (2) $\p{\phi} = \p{\psi}$ if $\phi=(\neg\psi)$; 
(3) $\p{\phi} = \p{\phi_1}\cup\p{\phi_2}$ if $\phi=(\phi_1\wedge\phi_2)$ or $(\phi_1\vee\phi_2)$.
\end{definition}

Consider $\phi= (a\wedge ((\neg Tail\wedge a) \U b)\wedge \neg (\neg Tail\wedge \X (a\vee b)))$. We have 
$\p{\phi}=\{a, Tail, ((\neg Tail\wedge a) \U b), (\X(a\vee b))\}$. Intuitively, the propositional atoms are obtained by treating all temporal subformulas of $\phi$ as atomic propositions. Thus,
an $\ltlf$ formula $\phi$ can be viewed as a propositional formula over $\p{\phi}$. 

\begin{definition}\label{def:pa}
For an $\ltlf$ formula $\phi$, let $\phi^p$ be $\phi$ considered as a propositional formula over $\p{\phi}$. A \emph{propositional 
assignment} $A$ of $\phi^p$, is in $2^{\p{\phi}}$ and satisfies $A\models\phi^p$. 
\end{definition}

Consider the formula $\phi = (a\vee (\neg Tail\wedge a) \U b) \wedge (b\vee (Tail\vee c) \R d)$. From Definition \ref{def:pa}, 
$\phi^p$ is 
$(a\vee p_1)\wedge (b\vee p_2)$ where $p_1$, $p_2$ are two Boolean variables representing the truth values of 
$(\neg Tail \wedge a) \U b$ and $(Tail \vee c) \R d$. 
Moreover, the set $\{\neg a, p_1 ((\neg Tail\wedge a) \U b), \neg b, p_2 ((Tail\vee c) \R d)\}$ is a propositional assignment 
of $\phi^p$. In the rest of 
the paper, we do not introduce the intermediate variables and directly say 
$\{\neg a, (\neg Tail \wedge a) \U b, \neg b, (Tail \vee c) \R d\}$ is a 
propositional assignment of $\phi^p$. The following theorem shows the relationship between the propositional assignment of 
$\phi^p$ and the satisfaction of $\phi$. 

\begin{theorem}\label{thm:assign}
For an $\ltlf$ formula $\phi$ and a finite trace $\xi$, $\xi\models\phi$ implies there exists a propositional assignment $A$ of $\phi^p$ such that $\xi\models\bigwedge A$; On the other hand, $\xi\models\bigwedge A$ where $A$ is a propositional assignment of $\phi^p$, also implies $\xi\models\phi$.
\end{theorem}

We now introduce the \textit{neXt Normal Form} (XNF) of $\ltlf$ formulas, which is useful for the construction of the transition system. 
\begin{definition}[neXt Normal Form]\label{def:xnf}
An $\ltlf$ formula $\phi$ is in \emph{neXt Normal Form} (XNF) 
if there are no 
Until or Release subformulas of $\phi$ in $\p{\phi}$. 
\end{definition}
For example, $\phi= ((\neg Tail \wedge a) \U b)$ is not in XNF,  while 
$(b\vee (\neg Tail\wedge a\wedge (\X((\neg Tail\wedge a) \U b))))$ is.
Every $\ltlf$ formula $\phi$ has a linear-time conversion to an equivalent formula in XNF, 
which we denoted as $\xnf{\phi}$.
\begin{theorem}\label{thm:xnf}
For an $\ltlf$ formula $\phi$, there is a corresponding $\ltlf$ formula $\xnf{\phi}$ in XNF such that 
$\phi\equiv\xnf{\phi}$.
Furthermore, the cost of the conversion is linear.
\end{theorem}

Observe that when $\phi$ is in XNF, there can be only Next (no Until or Release) temporal formulas in the propositional assignment 
of $\phi^p$. For $\phi = b \vee (a\wedge \neg Tail\wedge \X(a \U b))$, the set $A=\{a, \neg b, \neg Tail, \X(a \U b)\}$ 
is a propositional 
assignment of $\phi^p$. 
Based on $\ltlf$ semantics, we can induce from $A$ that if a finite trace $\xi$ satisfying $\xi[0]\supseteq \{a, \neg b, \neg Tail\}$ 
and $\xi_1\models a \U b$, $\xi\models \phi$ is true. This motivates us to construct the transition system for $\phi$, 
in which $\{a \U b\}$ is a next state 
of $\{\phi\}$ and $\{a, \neg b, \neg Tail\}$ is the transition label between these two states. 

Let $\phi$ be an $\ltlf$ formula and $A$ be a propositional assignment of $\phi^p$, we denote $L(A) =\{l | l\in A \textit{ is a literal}\}$ and  
$X(A) =\{\theta | \X\theta \in A\}$.  Now we define the 
\emph{transition system} for an $\ltlf$ formula.

\begin{definition}\label{def:ts}
Given an $\ltlf$ formula $\phi$ and its literal set $\mathcal{L}$, let $\Sigma = 2^{\mathcal{L}}$. We define the \emph{transition system} $T_{\phi}=(S,s_0,T)$ for $\phi$, where $S\subseteq 2^{cl(\phi)}$  is the set of states, $s_0 = \{\phi\}\in S$ is the \emph{initial state}, and
\begin{itemize}[noitemsep,topsep=0pt]
\item $T: S\times \Sigma \rightarrow 2^S$ is the transition relation, where $s_2\in T(s_1, \sigma)$ ($\sigma\in \Sigma$) holds iff there is a propositional assignment 
$A$ of $\xnf{\bigwedge s_1}^p$ such that $\sigma \supseteq L(A)$ and $s_2 = X(A)$. 
\end{itemize}
A \emph{run} of $T_\phi$ on a finite trace $\xi (|\xi|=n>0)$ is a finite sequence $s_0,s_1,\ldots, s_n$ such that $s_0$ is the initial state and $s_{i+1}\in T(s_i,\xi[i])$ holds for all $0\leq i < n$.
\end{definition}

We define the notation $|r|$ for a run $r$, to represent the length of $r$, i.e. number of states in $r$. We say state $s_2$ is reachable from state $s_1$ in $i(i\geq 0)$ steps (resp. in up to $i$ steps), if there is a run $r$ on some finite trace $\xi$ leading from $s_1$ to $s_2$ and $|r| = i$ (resp. $|r|\leq i$). In particular, we say $s_2$ is a \emph{one-transition next state} of $s_1$ if $s_2$ is reachable from $s_1$ in 1 steps. Since a state $s$ is a subset of $cl(\phi)$, which essentially is a formula with the form of $\bigwedge_{\psi\in s}\psi$, we mix the usage of the state and formula in the rest of the paper. That is, a state can be a formula of $\bigwedge_{\psi\in s}\psi$, and a formula $\phi$ can be a set of states, i.e. $s\in\phi$ iff $s\Rightarrow \phi$.

\begin{lemma}\label{lem:reachable}
Let $T_{\phi}=(S,s_0,T)$ be the transition system of $\phi$. Every state $s\in S$ is reachable from the initial state $s_0$.
\end{lemma}

\begin{definition}[Final State]\label{def:final}
Let $s$ be a state of a transition system $T_{\phi}$. Then $s$ is a \emph{final state} of $T_{\phi}$ iff the Boolean formula $Tail\wedge (\xnf{s})^p$ is satisfiable.
\end{definition}

By introducing the concept of \emph{final state}, we are able to check the satisfiability of the $\ltlf$ formula $\phi$ over $T_{\phi}$.

\begin{theorem}\label{thm:reasoning}
Let $\phi$ be an $\ltlf$ formula. Then $\phi$ is satisfiable  iff there is a final state in $T_{\phi}$. 
\end{theorem}

An intuitive solution from Theorem \ref{thm:reasoning} to check the satisfiability of $\phi$ is to construct states of $T_{\phi}$ until (1) either a final state is found by Definition \ref{def:final}, meaning $\phi$ is satisfiable; or (2) all states in $T_{\phi}$ are generated but no final state can be found, meaning $\phi$ is unsatisfiable. This approach is simple and easy to implement, however, it does not perform well according to our preliminary experiments. 

%% file: cdlsc.tex
\section{Conflict-Driven $\ltlf$ Satisfiability Checking}

In this section, we present a conflict-driven algorithm for $\ltlf$ satisfiability checking. The new algorithm is inspired by \cite{LZPV15}, where information of both satisfiability and unsatisfiability results of  SAT solvers are used. The motivation is as follows:
In Definition \ref{def:final}, if the Boolean formula $Tail\wedge \xnf{s}^p$ is unsatisfiable, the SAT solver is able to provide a UC (Unsatisfiable Core) $c$ such that $c\subseteq s$ and $Tail\wedge\xnf{c}^p$ is still unsatisfiable. It means that $c$ represents a set of states that are not final states. By adding a new constraint $\neg (\bigwedge_{\psi\in c}\X\psi)$, the SAT solver 
can provide a model (if exists) that avoids re-generation of those states in $c$, which accelerates the search of final states. 
More generally, we define the \emph{conflict sequence}, which is used to maintain all information of UCs acquired during the checking process.

\begin{definition}[Conflict Sequence]\label{def:cs}
Given an $\ltlf$ formula $\phi$, a conflict sequence $\cs$ for the transition system $T_{\phi}$ is a finite 
sequence of set of states such that:
\begin{enumerate}[noitemsep,topsep=0pt]
    \item The initial state $s_0=\{\phi\}$ is in $\cs[i]$ for $0 \leq i< |\cs|$;
    \item Every state in $\cs[0]$ is not a final state;
    \item For every state $s\in\cs[i+1]$ ($0\leq i<|\cs|-1$), all the one-transition next states of $s$ are included in $\cs[i]$. 
\end{enumerate}
We call each $\cs[i]$ is a \emph{frame}, and $i$ is the  \emph{frame level}.

\end{definition}

In the definition, $|\cs|$ represents the length of $\cs$ and $\cs[i]$ denotes the $i$-th element of $\cs$. 
Consider the transition system shown in Figure \ref{fig:cs}, in which $s_0$ is the initial state and $s_4$ is the final state. 
Based on Definition \ref{def:cs}, the sequence $\cs=\{s_0,s_1,s_2,s_3\}, \{s_0,s_1\}, \{s_0\}$ is a conflict sequence.
Notably, the conflict sequence for a transition system may not be unique. For the above example, the sequence 
$\{s_0,s_1\}, \{s_0\}$ is also a conflict sequence for the system. 
This suggests that the construction of a conflict sequence is algorithm-specific.
Moreover, it is not hard to induce that every non-empty prefix of a conflict sequence is also a conflict sequence. 
For example, a prefix of $\cs$ above, i.e. $\{s_0,s_1,s_2,s_3\}, \{s_0, s_1\}$, is a conflict sequence. 
As a result, a conflict sequence can be constructed iteratively, i.e. the elements can be generated (and updated) in order. 
Our new algorithm is motivated by these two observations.

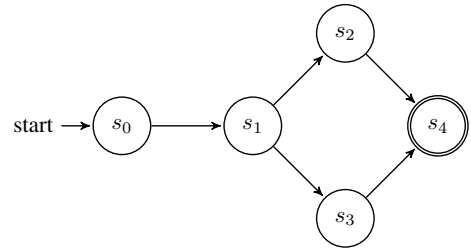
\begin{figure}[!htb]
\centering
\scalebox{0.87}{
\begin{tikzpicture}[->,>=stealth',shorten >=1pt,auto,node distance=2cm,
                    semithick]
  \tikzstyle{every state}=[fill=none,draw=black,text=black]

  \node[state,initial] 		 (A)              {$s_0$};
  \node[state]               (B) [right of=A] {$s_1$};
  \node[state]               (C) [above right of=B] {$s_2$};
  \node[state]               (D) [below right of=B] {$s_3$};
  \node[state,accepting]         (E) [above right of=D] {$s_4$};

  \path (A) edge              node {} (B)
        (B) edge 			  node {} (C)
        (B) edge 			  node {} (D)
        (C) edge              node {} (E)
        (D) edge              node {} (E);

\end{tikzpicture}
}
\caption{An example transition system for the conflict sequence.}\label{fig:cs}
\end{figure}

An inherent property of conflict sequences is described in the following lemma. 
\begin{lemma}\label{lem:cs}
Let $\phi$ be an $\ltlf$ formula with a conflict sequence $\cs$ for the transition system $T_{\phi}$, then $\bigcap_{0\leq j\leq i}\cs[j] (0\leq i< |\cs|)$ represents a set of states that cannot reach a final state in up to $i$ steps. 
\end{lemma}
\begin{proof}
    We first prove $\cs[i](i\geq 0)$ is a set of states that cannot reach a final state in $i$ step. 
    Basically from Definition \ref{def:cs}, 
    $\cs[0]$ is a set of states that are not final states. Inductively, assume $\cs[i](i\geq 0)$ is a set of states that cannot reach 
    a final state in $i$ steps. From Item 3 of Definition \ref{def:cs}, every state $s\in\cs[i+1]$ satisfies all its one-transition next states are in $\cs[i]$, thus every state $s\in\cs[i+1]$ cannot reach a final state in $i+1$ steps. 
    Now since $\cs[i](i\geq 0)$ is a set of states that cannot reach a final state in $i$ steps, $\bigcap_{0\leq j\leq i}\cs[j]$ is a set of states that cannot reach a final state in up to $i$ steps.
\end{proof}


We are able to utilize the conflict sequence to accelerate the satisfiability checking of $\ltlf$ formulas, using the theoretical foundations provided by Theorem \ref{thm:sat} and \ref{thm:unsat} below.

\begin{theorem}\label{thm:sat}
    The $\ltlf$ formula $\phi$ is satisfiable iff there is a run $r=s_0, s_1,\ldots, s_n(n\geq 0)$ of $T_{\phi}$ such that 
    (1) $s_n$ is a final state; and (2) $s_i$ ($0\leq i \leq n$) is not in $\cs[n-i]$ for every conflict sequence $\cs$ of $T_{\phi}$ 
    with $|\cs|> n-i$.
\end{theorem}
\begin{proof}
($\Leftarrow$) Since $s_n$ is a final state, $\phi$ is satisfiable according to Theorem \ref{thm:reasoning}. 
($\Rightarrow$) Since $\phi$ is satisfiable, there is a finite trace $\xi$ such that the corresponding run $r$ of $T_{\phi}$ on $\xi$ ends with a final state (according to Theorem \ref{thm:reasoning}). Let  $r$ be $s_0\tran{}s_1\tran{}\ldots s_{n}$ where $s_{n}$ is the final state. It holds that $s_i$ ($0\leq i\leq n$) is a state that can reach a final state in $n-i$ steps. 
Moreover for every $\cs$ of $T_{\phi}$ with $|\cs|>n-i$, $\cs[n-i]$ ($\cs[n-i]$ is meaningless when $|\cs|\leq n-i$) 
represents a set of states that cannot 
reach a final state in $n-i$ steps (From the proof of Lemma \ref{lem:cs}). As a result, it is true that $s_i$ is not in $\cs[n-i]$ 
if $|\cs|>n-i$.
\end{proof}

Theorem \ref{thm:sat} suggests that to check whether a state $s$ can reach a final state in $i$ steps ($i\geq 1$), finding a 
one-transition next state $s'$ of $s$ that is not in $\cs[i-1]$ is necessary; as $s'\in\cs[i-1]$ imples 
$s'$ cannot reach a final state in $i-1$ steps (From the proof of Lemma \ref{lem:cs}). If all one-transition next states of $s$ are 
in $\cs[i-1]$, $s$ cannot reach a final state in $i$ steps.

\begin{theorem}\label{thm:unsat}
    The $\ltlf$ formula $\phi$ is unsatisfiable iff there is a conflict sequence $\cs$ and $i\geq 0$ such that $\bigcap_{0\leq j\leq i}\cs[j] \subseteq \cs[i+1]$.
\end{theorem}
\begin{proof}
($\Leftarrow$) $\bigcap_{0\leq j\leq i}\cs[j] \subseteq \cs[i+1]$ is true implies that $\bigcap_{0\leq j\leq i}\cs[j] = \bigcap_{0\leq j\leq i+1}\cs[j]$ is true. Also from Lemma \ref{lem:cs} we know $\bigcap_{0\leq j\leq i}\cs[j]$ is a set of states that cannot reach a final state in up to i steps. Since $\phi\in \cs[i]$ is true for each $i\geq 0$, $\phi$ is in $\bigcap_{0\leq j\leq i}\cs[j]$. 
Moreover, $\bigcap_{0\leq j\leq i}\cs[j] = \bigcap_{0\leq j\leq i+1}\cs[j]$ is true implies all reachable states from $\phi$ are included in $\bigcap_{0\leq j\leq i}\cs[j]$. We have known all states in $\bigcap_{0\leq j\leq i}\cs[j]$ are not final states, so $\phi$ is unsatisfiable. 

($\Rightarrow$) If $\phi$ is unsatisfiable, every state in $T_{\phi}$ is not a final state. Let $S$ be the set of states of $T_{\phi}$. According to Lemma \ref{lem:cs}, $\bigcap_{0\leq j\leq i}\cs[j] (i\geq 0)$ contains the set of states that are not final in up to $i$ steps. Now we let $\cs$ satisfy that  $\bigcap_{0\leq j\leq i}\cs[j] (i\geq 0)$ contains all states that are not final in up to $i$ steps, so $\bigcap_{0\leq j\leq i}\cs[j]$ includes all reachable states from $\phi$, as $\phi$ is unsatisfiable. However, because $\bigcap_{0\leq j\leq i}\cs[j] \supseteq \bigcap_{0\leq j\leq i+1}\cs[j]\supseteq S (i\geq 0)$, there must be an $i\geq 0$ such that $\bigcap_{0\leq j\leq i}\cs[j] = \bigcap_{0\leq j\leq i+1}\cs[j]$, which indicates that $\bigcap_{0\leq j\leq i}\cs[j]\subseteq \cs[i+1]$ is true. 
\end{proof}

\noindent\textbf{Algorithm Design.} The algorithm, named \cdlsc (Conflict-Driven $\ltlf$ Satisfiability Checking), 
constructs the transition system on-the-fly. The initial state $s_0$ 
is fixed to be $\{\phi\}$ where $\phi$ is the input formula. From Definition \ref{def:final}, whether a state $s$ is final 
is reducible to the satisfiability checking of the Boolean formula $Tail\wedge \xnf{s}^p$. If $s_0$ is a final state, 
there is no need to maintain the conflict sequence in \cdlsc, and the algorithm can return SAT immediately; Otherwise, the 
conflict sequence is maintained as follows.
\begin{itemize}
    \item In \cdlsc, every element of $\cs$ is a set of set of subformulas of the input formula $\phi$. 
    Formally, each $\cs[i]$ ($i\geq 0$) can be represented by the $\ltlf$ formula $\bigvee_{c\in \cs[i]}\bigwedge_{\psi\in c} \psi$ 
    where $c$ is 
    a set of subformulas of $\phi$. We mix-use the notation $\cs[i]$ for the corresponding $\ltlf$ formula as well. 
    Every state $s$ satisfying $s\Rightarrow\cs[i]$ is included in $\cs[i]$. 
    \item $\cs$ is created iteratively. In each iteration $i\geq 0$, $\cs[i]$ is initialized as the empty set.
    \item To compute elements in $\cs[0]$, we consider an existing state $s$ (e.g. $s_0$). If the Boolean formula 
    $Tail\wedge\xnf{s}^p$ is unsatisfiable, $s$ is not a final state and can be added into $\cs[0]$ from Item 2 of 
    Definition \ref{def:cs}. 
    Moreover, \cdlsc leverages the Unsatisfiable Core (UC) technique from the SAT community to add a set of states, all of which are not 
    final and include $s$, to $\cs[0]$. This set of states, denoted as $c$, 
    is also represented by a set of $\ltlf$ formulas and satisfies $c\subseteq s$. 
    The detail to obtain $c$ is discussed below. 
    \item To compute elements in $\cs[i+1]$ ($i\geq 0$), we consider the Boolean formula $(\xnf{s}\wedge\neg \X(\cs[i]))^p$, 
    where $\X(\cs[i])$ represents the $\ltlf$ formula $\bigvee_{c\in\cs[i]}\bigwedge_{\psi\in c}\X(\psi)$.  
    The above Boolean formula is used to check whether there is a one-transition next state of $s$ that is not in $\cs[i]$. 
    If the formula is unsatisfiable, all the one-transition next states of $s$ are in $\cs[i]$, thus $s$ can be added into 
    $\cs[i+1]$ according to Item 3 of Definition \ref{def:cs}. Similarly, we also utilize the UC technique to obtain a subset $c$ of 
    $s$, such that $c$ represents a set of states that can be added into $\cs[i+1]$. 
\end{itemize} 

As shown above, every Boolean formula sent to a SAT solver has the form of ($\xnf{s}\wedge \theta)^p$ where $s$ is a state and 
$\theta$ is either $Tail$ or $\neg \X(\cs[i])$. 
Since every state $s$ consists of a set of $\ltlf$ formulas, the Boolean formula can be rewritten as 
$\alpha_1=(\bigwedge_{\psi\in s}\xnf{\psi}\wedge \theta)^p$. Moreover, we introduce a new Boolean variable $p_{\psi}$ for each 
$\psi\in s$, and re-encode the formula to be 
$\alpha_2=\bigwedge_{\psi\in s}p_{\psi} \wedge (\bigwedge_{\psi\in s}(\xnf{\psi}\vee \neg p_{\psi})\wedge \theta)^p$. 
$\alpha_2$ is satisfiable iff $\alpha_1$ is satisfiable, and $A$ is an assignment of $\alpha_2$ iff 
$A\backslash \{p_{\psi}|\psi\in s\}$ is an assignment of $\alpha_1$. Sending $\alpha_2$ instead of $\alpha_1$ to the SAT solver that 
supports assumptions (e.g. Minisat \cite{ES03}) enables the SAT solver to return the UC, which is a set of $s$, when $\alpha_2$ is 
unsatisfiable. For example, assume $s= \{\psi_1, \psi_2, \psi_3\}$ and $\alpha_2$ is sent to the SAT solver with 
$\{p_{\psi_i}|i\in \{1, 2, 3\}\}$ being the assumptions. If the SAT solver returns unsatisfiable and the UC $\{p_{\psi_1}\}$, 
the set $c=\{\psi_1\}$, which represents every state including $\psi_1$, is the one to be added into the corresponding $\cs[i]$. 
We use the notation $get\_uc()$ for the above procedure. 

The pseudo-code of \cdlsc is shown in Algorithm \ref{alg:cdlsc}.  
Line \ref{alg:cdlsc:tailstart}-\ref{alg:cdlsc:tailend} considers the situation when the input formula $\phi$ is a final state itself. Otherwise, the first frame $\cs[0]$ is initialized to $\{\phi\}$ (Line \ref{alg:cdlsc:csinit}), and the current frame level is set to 0 (Line \ref{alg:cdlsc:framelevelinit}). After that, the loop body (Line \ref{alg:cdlsc:loopstart}-\ref{alg:cdlsc:loopend}) keeps updating the elements of $\cs$ iteratively, until either the procedure $try\_satsify$ returns true, which means to find a model of $\phi$, or the procedure $inv\_found$ returns true, which is the implementation of Theorem \ref{thm:unsat}. The loop continues to create a new frame in $\cs$ if neither of the procedures succeeds to return true. To describe conveniently, we say every run of the while loop body in Algorithm \ref{alg:cdlsc} is an \emph{iteration}.

\begin{algorithm}
\caption{Implementation of \cdlsc}\label{alg:cdlsc}
  \begin{algorithmic}[1]      
       \REQUIRE An $\ltlf$ formula $\phi$.
       \ENSURE  SAT or UNSAT.
       \IF{$Tail\wedge \xnf{\phi}^p$ is satisfiable}\label{alg:cdlsc:tailstart}
            \RETURN SAT;
        \ENDIF\label{alg:cdlsc:tailend}
       \STATE Set $\cs[0] := \{\phi\}$;\label{alg:cdlsc:csinit}
       \STATE Set $frame\_level := 0$;\label{alg:cdlsc:framelevelinit}
       \WHILE{true}\label{alg:cdlsc:loopstart}
            \IF{$try\_satisfy (\phi,frame\_level)$ returns true}\label{alg:cdlsc:trysatisfy}
                \RETURN SAT;
            \ENDIF
            \IF{$inv\_found (frame\_level)$ returns true}
                \RETURN UNSAT;
            \ENDIF
            \STATE $frame\_level := frame\_level+1$;
            \STATE Set $\cs[frame\_level] = \emptyset$;
       \ENDWHILE\label{alg:cdlsc:loopend}
     \end{algorithmic}
\end{algorithm}%

The procedure $try\_satisfy$ is responsible for updating $\cs$. Taking an formula $\phi$ and the frame level $frame\_level$ currently working on, $try\_satisfy$ returns true iff a model of $\phi$ can be found, with the length of $frame\_level+1$. As shown in Algorithm  \ref{alg:satisfy}, $try\_satisfy$ is implemented in a recursive way. Each time it checks whether a next state of the input $\phi$, which belongs to a lower level (than the input $frame\_level$) frame can be found (Line \ref{alg:satisfy:loopstart}). If the result is positive and such a new state $\phi'$ is constructed, $try\_satisfy$ first checks whether $\phi'$ is a final state when $frame\_level$ is 0 (in which case returns true). If $\phi'$ is not a final state, a UC is extracted from the SAT solver and added to $\cs[0]$ (Line \ref{alg:satisfy:frame0start}-\ref{alg:satisfy:frame0end}). If $frame\_level$ is not 0, $try\_satisfy$ recursively checks whether a model of $\phi'$ can be found with the length of $frame\_level$ (Line \ref{alg:satisfy:recursivestart}-\ref{alg:satisfy:recursiveend}). If the result is negative and such a state cannot be constructed, a UC is extracted from the SAT solver and added into $\cs[frame\_level+1]$ (Line \ref{alg:satisfy:ucstart}-\ref{alg:satisfy:ucend}). 

%
\begin{algorithm}
\caption{Implementation of $try\_satisfy$}\label{alg:satisfy}
  \begin{algorithmic}[1]      
       \REQUIRE $\phi$: The formula is working on;\\$frame\_level$: The frame level is working on.
       \ENSURE  true or false.
       \STATE Let $\psi = \neg \X(\cs[frame\_level])$;
       \WHILE {$(\psi\wedge\xnf{\phi})^p$ is satisfiable}\label{alg:satisfy:loopstart}
            
            \STATE Let $A$ be the model of $(\psi\wedge\xnf{\phi})^p$;
            \STATE Let $\phi' = X(A)$, i.e. be the next state of $\phi$ extracted from $A$;
            \IF{$frame\_level == 0$}\label{alg:satisfy:frame0start}
                \IF{ $Tail\wedge \xnf{\phi'}^p$ is satisfiable}\label{alg:satisfy:satstart}
                    \RETURN true; \label{alg:satisfy:satend}
                \ELSE
                    \STATE Let $c = get\_uc ()$;\label{alg:satisfy:uc0start}
                    \STATE Add $c$ into $\cs[frame\_level]$;\label{alg:satisfy:uc0end}
                    \STATE Continue;
                \ENDIF
            \ENDIF\label{alg:satisfy:frame0end}
            \IF{$try\_satisfy (\phi', frame\_level-1)$ is true}\label{alg:satisfy:recursivestart}
                \RETURN true;
            \ENDIF\label{alg:satisfy:recursiveend}
       \ENDWHILE
       \STATE Let $c = get\_uc ()$;\label{alg:satisfy:ucstart}
       \STATE Add $c$ into $\cs[frame\_level+1]$;\label{alg:satisfy:ucend}
       \RETURN false;
     \end{algorithmic}
\end{algorithm}

Notably, Item 1 of Definition \ref{def:cs}, i.e. $\{\phi\}\in\cs[i]$, is guaranteed for each $i\geq 0$, 
as the original input formula of $try\_satisfy$ is always $\phi$ (Line \ref{alg:cdlsc:trysatisfy} in Algorithm \ref{alg:cdlsc}) 
and there is some $c$ (Line \ref{alg:satisfy:ucend} in Algorithm \ref{alg:satisfy}) 
including $\{\phi\}$ that is added into $\cs[i]$, if no model can be found in the current iteration.

The procedure $inv\_found$ in Algorithm \ref{alg:cdlsc} implements Theorem \ref{thm:unsat} in a straightforward way: It reduces the checking of whether $\bigcap_{0\leq j\leq i}\cs[j] \subseteq \cs[i+1]$  being true on some frame level $i$, to the satisfiability checking of the Boolean formula $\bigwedge_{1\leq j\leq i} \cs[j] \Rightarrow \cs[i+1]$.
Finally, we  state Theorem \ref{thm:terminate} below to provide the theoretical 
guarantee that \cdlsc always terminates correctly.

\begin{lemma}\label{lem:cdlsc}
After each iteration of \cdlsc with no model found, the sequence $\cs$ is a conflict sequence of $T_{\phi}$ for the transition system 
$T_{\phi}$. 
\end{lemma}

\begin{theorem}\label{thm:terminate}
The \cdlsc algorithm terminates with a correct result.
\end{theorem}

 \cdlsc is shown how to accelerate the checking of satisfiable formulas in the previous section. 
 For unsatisfiable instances, consider 
 $\phi = (\neg Tail) \U a \wedge (Tail)\R\neg a\wedge (\neg Tail)\U b$. 
 \cdlsc first checks that $Tail\wedge \xnf{\phi}^p$ is unsatisfiable, where the SAT solver returns 
 $c=\{(\neg Tail)\U a, Tail\R \neg a\}$ as the UC. 
 So $c$ is added into $\cs[0]$. Then \cdlsc checks that $(\xnf{\phi}\wedge \neg \X(\cs[0]))^p$ is still unsatisfiable, 
 in which $c=\{(\neg Tail)\U a, Tail\R \neg a\}$ is still the UC. 
 So $c$ is added into $\cs[1]$ as well. Since $\cs[0]\subseteq \cs[1]$ and according to Theorem \ref{thm:unsat}, \cdlsc terminates with  
 the unsatisfiable result. In this case, \cdlsc only visits one state for the whole checking process. For a more general instance 
 like $\phi\wedge \psi$, where $\psi$ is a large $\ltlf$ formula, checking by \cdlsc enables to achieve a significantly improvement 
 compared to the checking by traditional tableau approach. 
 
 Summarily, \cdlsc is a conflict-driven on-the-fly satisfiability checking algorithm, which successfully leads to either an earlier finding of a satisfying model, or the faster termination with the unsatisfiable result.

%% file: experiment.tex
\section {Experimental Evaluation}\label {sec:exp}



\begin{figure}
\centering
\vspace*{-0.4cm}  
\includegraphics[scale=0.7]{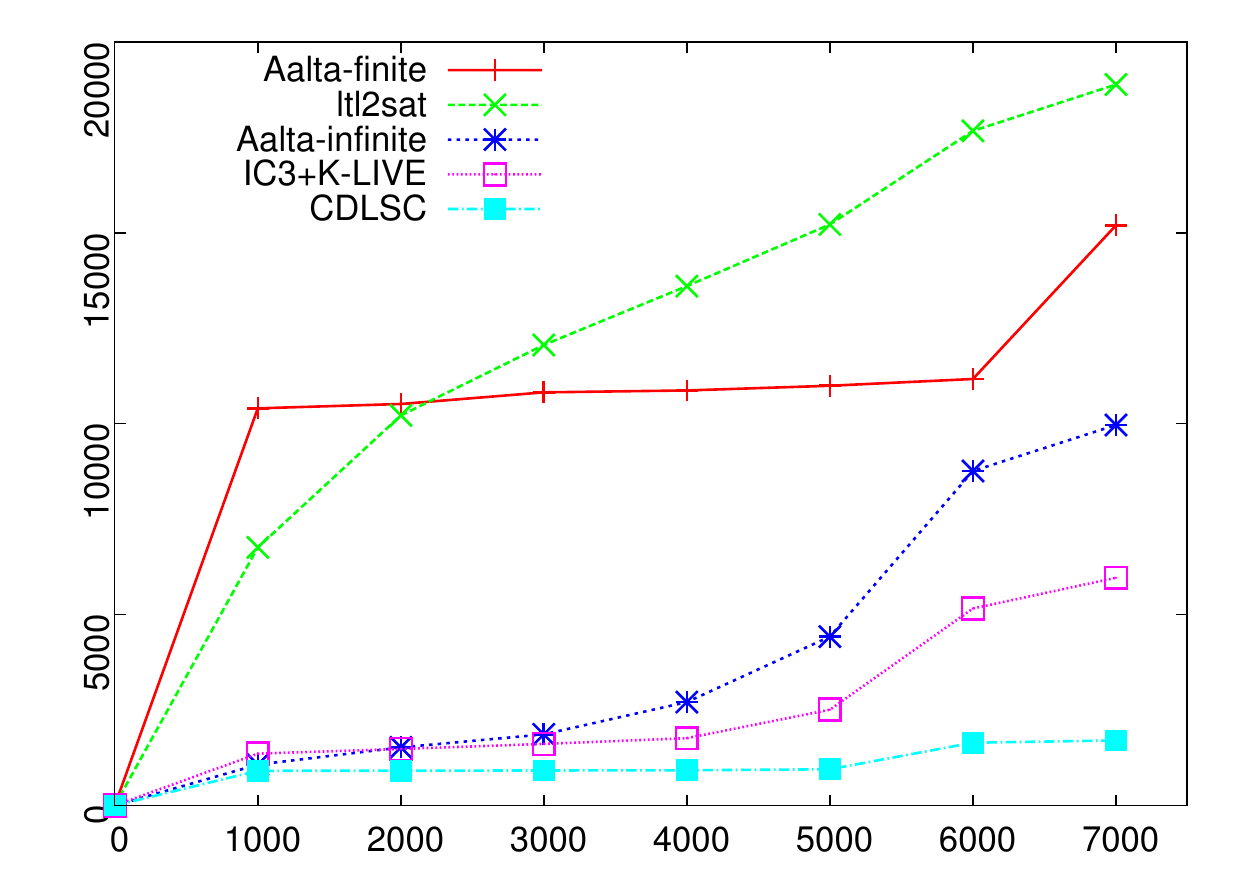}
\caption{Result for $\ltlf$ Satisfiability Checking on LTL-as-$LTL_f$ Benchmarks. The X axis represents the number of benchmarks, and the Y axis is the accumulated checking time (s).}
\label{fig:cactus}
\end{figure}

\begin{table*}[!htb]
\centering
\scalebox{0.8}
{
\begin{tabular}{lcccccccccc}
\hline
Type & Number & Result & IC3+K-LIVE & Aalta-finite & Aalta-infinite & ltl2sat & \cdlsc \\
\hline
Alternate Response &  100 & sat & 134 & 1 & 48 & 123 & 3 \\
Alternate Precedence & 100 & sat & 154 & 3 & 70 & 380 & 4 \\
Chain Precedence &  100 & sat & 127 & 2 & 45 & 83 & 2 \\
Chain Response & 100 & sat & 79 & 1 & 41 & 49 & 2 \\
Precedence & 100 &  sat & 132 & 2 & 14 & 124 & 1 \\
Responded Existence  & 100 & sat & 130 & 1 & 14 & 327 & 1 \\
Response & 100 & sat & 155 & 1 & 41 & 53 & 2 \\
Practical Conjunction & 1000 & varies & 1669 & 19564 & 4443 & 20477 & 115 \\
\hline
\end{tabular}
}
\caption{Results for $\ltlf$ Satisfiability Checking on $LTL_f$-specific Benchmarks.}
\label{tab:ltlf-specific}
\end{table*}

\noindent\textbf{Benchmarks}
We first consider the \emph{LTL-as-$\ltlf$} benchmark, which is evaluated by previous works on $\ltlf$ satisfiability checking \cite{LZPVH14,FG16}. This benchmark consists of 7442 instances that are originally LTL formulas but are treated as $\ltlf$ formulas, 
as both logics share the same syntax. Previous works \cite{LZPVH14,FG16} have shown that the benchmark is useful to test 
the scalability of $\ltlf$ solvers.

Secondly, we consider the 7 \emph{$\ltlf$-specific} patterns that are introduced in recent researches on $\ltlf$, e.g. 
\cite{DMM14,DMM16}, and we create 100 instances for each pattern. As shown in Table \ref{tab:ltlf-specific}, 
it is trivial to check the satisfiability of these $\ltlf$ patterns 
by most tested solvers, as either they have small sizes or dedicated heuristics for $\ltlf$, which are encoded in both Aalta-finite 
and \cdlsc, enable to solve them quickly. Inspired from the observation in \cite{LZPVH13} that an LTL specification 
in practice is often the conjunction of a set of small and frequently-used patterns, we randomly choose a subset of the instances of the 7 patterns to imitate a real $\ltlf$ specification in practice. We generate 1000 such instances as the \emph{practical conjunction} pattern shown in the last row of Table \ref{tab:ltlf-specific}. Unlike the random benchmarks in SAT community, which are often considered not interesting, we argue that the new practical conjunction pattern is a representative for real $\ltlf$ specifications in industry. 

\noindent\textbf{Experimental Setup}
We implement  \cdlsc in C++, and use Minisat 2.2.0 \cite{ES03} as the SAT engine\footnote{https://github.com/lijwen2748/aaltaf}. We compare it with two extant $\ltlf$ satisfiability solvers: Aalta-finite \cite{LZPVH14} and ltl2sat \cite{FG16}. We also compared with the state-of-art LTL solver Aalta-infinite \cite{LZPV15}, using the $\ltlf$-to-LTL satisfiability-preserving reduction described in \cite{GV13}. As LTL satisfiability checking is reducible to model checking, as described in \cite{RV07}, we also compared with this reduction, using nuXmv with the IC3+K-LIVE back-end \cite{CCDGMMMRT14}, 
as an $\ltlf$ satisfiability checker. 

We ran the experiments on a RedHat 6.0 cluster with 2304 processor cores
in 192 nodes (12 processor cores per node), running at 2.83 GHz with 48GB of
RAM per node. 
Each tool executed on a dedicated node with a timeout of 60 seconds, measuring execution time with Unix \texttt{time}. Excluding timeouts, all solvers found correct verdicts for all formulas. All artifacts are available in the supplemental material.

\noindent\textbf{Results}
Figure \ref{fig:cactus} shows the results for $\ltlf$ satisfiability
checking on LTL-as-$\ltlf$ benchmarks. \cdlsc outperforms all other
approaches. On average, \cdlsc performs about 4 times faster than
the second-best approach IC3+K-LIVE (1705 seconds vs. 6075 seconds). 
\cdlsc checks the $\ltlf$ formula directly, while IC3+K-LIVE must take the input of 
the LTL formula translated from the $\ltlf$ formula. As a result, IC3-KLIVE may take extra 
cost, e.g. finding a satisfying lasso for the model, to the satisfiability checking. 
Meanwhile, \cdlsc can benefit from the heuristics dedicated for $\ltlf$ that are proposed in \cite{LZPVH14}.
Finally, the performance of ltl2sat is highly tied to 
its performance of unsatisfiability checking as most of the timeout cases for ltl2sat are unsatifiable. 
For Aalta-finite, its performance is restricted by the heavy cost of Tableau 
Construction.

Table \ref{tab:ltlf-specific} shows the results for $\ltlf$-specific experiments. Columns 1-3 show the types of $\ltlf$ formulas under test, the number of test instances for each formula type, and the results by formula type. Columns 4-8 show the checking times by formula types in seconds. The dedicated $\ltlf$ solvers perform extremely fast on the seven scalable pattern formulas (Column 5 and 8), because their heuristics work well on these patterns. For the difficult conjunctive benchmarks, \cdlsc still outperforms all other solvers. 
 

%% file: conclusion.tex
\section{Discussion and Concluding Remarks}\label{sec:con}
Bounded Model Checking (BMC) \cite{BCCZ99} is also a popular SAT-based technique, which is however, not necessary to compare. 
There are two ways to apply BMC to $\ltlf$ satisfiability checking. The first one is to check the satisfiability of the LTL formula 
from the input $\ltlf$ formula. \cite{LZPV15} has shown that this approach cannot perform better than IC3+K-LIVE, and the fact of 
\cdlsc outperforming IC3+K-LIVE induces \cdlsc also outperforms BMC. 
The second approach is to check the satisfiability 
of the $\ltlf$ formula $\phi$ directly, by unrolling $\phi$ iteratively. In the worst case, BMC can terminate (with UNSAT) once 
the iteration reaches the upper bound. This is exactly what is implemented in ltl2sat \cite{FG16}. 

In this paper, we introduce a new SAT-based framework, 
based on which we present a conflict-driven algorithm \cdlsc, for $\ltlf$ satisfiability checking.
Our experiments demonstrate that \cdlsc outperforms Aalta-infinite and IC3+K-LIVE, which are designed for LTL satisfiability checking, showing the advantage of a dedicated algorithm for $\ltlf$. Notably, \cdlsc maintains a conflict sequence, which is similar to the state-of-art model checking technique IC3 \cite{Bra11}. 
\cdlsc does not require the conflict sequence to be monotone, and simply use the UC from SAT solvers to update the sequence. 
Meanwhile, IC3 requires the sequence to be strictly monotone, and has to compute its dedicated MIC (Minimal Inductive Core) to update 
the sequence.  
%
We conclude that \cdlsc outperforms other existing approaches for $\ltlf$ satisfiability checking. 

%% file: append.tex
\newpage
\appendix 

\section{Missing Proofs}
\subsection{Proof of Theorem \ref{thm:tnf}}

We first introduce the following lemmas that are useful for the proof. 
\begin{lemma}\label{lem:tnf1}
If $\tnf{\phi}$ is satisfiable, there is a non-empty finite trace $\xi$ such that $\neg Tail\in \xi[i]$ for $0\leq i<|\xi|-1$, 
$Tail \in \xi[|\xi|-1]$ and $\xi\models\tnf{\phi}$.
\end{lemma}
\begin{proof}
Since $\tnf{\phi}$ is satisfiable, there is a non-empty finite trace $\xi'$ such that $\xi'\models\tnf{\phi}$. 
Recall that $\tnf{\phi}$ has the form of $t(\phi)\wedge FTail$, so $\xi'\models\tnf{\phi}$ implies $\xi'\models t(\phi)$ and 
there is $\leq k<|\xi'|$ such that 
$Tail\in\xi'[k]$ and $Tail\not\in \xi'[j]$ for every $j< k$. We define $tp(\xi') = \xi'[0]\xi'[1]\ldots\xi'[k]$, and first 
prove that $\xi'\models t(\phi)$ implies $tp(\xi')\models t(\phi)$. Let 
$\xi = tp(\xi')$, and we prove by induction over the type of $\phi$ that $\xi\models t(\phi)$.
\begin{enumerate}
    \item If $\phi = \tt$, then $t(\phi)=\tt$ and of course $\xi\models t(\phi)$;

    \item If $\phi = l$ is a literal, then $t(\phi) = l$ and $\xi'\models t(\phi)$ implies $l\in \xi'[0] = \xi[0]$. Therefore, 
    $\xi\models t(\phi)$;
    
    \item If $\phi = \phi_1\wedge \phi_2$, then $t(\phi) = t(\phi_1)\wedge t(\phi_2)$, and $\xi'\models t(\phi)$ implies 
    $\xi'\models t(\phi_1)$ and $\xi'\models t(\phi_2)$. By hypothesis assumption, $\xi'\models t(\phi_1)$ implies 
    $\xi\models t(\phi_1)$ and $\xi'\models t(\phi_2)$ implies
    $\xi\models t(\phi_2)$. So $\xi\models t(\phi)$ is true. If $\phi = \phi_1 \vee \phi_2$, the proof is similar; 
    
    \item If $\phi = X\psi$, then $t(\phi) = \neg Tail \wedge X (t(\psi))$, and $\xi'\models t(\phi)$ implies that $Tail\not\in\xi'[0]$ 
    and $\xi'_1\models t(\psi)$. Let $\xi_1 = tp (\xi'_1)$. By hypothesis assumption, $\xi'_1\models t(\psi)$ implies 
    $\xi_1\models t(\psi)$ is true. Moreover, because 
    $Tail\not\in\xi'[0]$, $\xi = tp (\xi') = \xi'[0]\cdot tp (\xi'_1) = \xi'[0]\cdot\xi_1$ from its definition. 
    As a result, $\xi\models t(\phi)$ is true;
    
    \item If $\phi = N\psi$, then $t(\phi) = Tail \vee X(t(\psi)) = Tail \vee (\neg Tail \wedge X(t(\psi)))$, 
    and $\xi'\models t(\phi)$ implies that $Tail\in \xi'[0]$ or 
    $\xi'\models \neg Tail\wedge X(t(\psi))$. In the first case, $\xi = \xi'[0]$ and obviously $\xi\models t(\phi)$. 
    For the second case, the proof is the same as that if $\phi = X\psi$;
    
    \item If $\phi = \phi_1 U\phi_2$, then $t(\phi) = (\neg Tail\wedge t(\phi_1)) U t(\phi_2)$, and $\xi'\models t(\phi)$ implies that 
    there is $0\leq i<|\xi'|$ such that $\xi'_i\models t(\phi_2)$ and for every $0\leq j<i$ it holds 
    $\xi'_j\models \neg Tail\wedge t(\phi_1)$. As a result, we have that $\xi = tp (\xi') = \xi'[0]\ldots\xi'[i-1]\cdot tp(\xi'_i)$, 
    and thus 
    $\xi_i = tp(\xi'_i)$ and $\xi_j = \xi'[j]\ldots\xi'[i-1]\cdot tp(\xi'_i) = tp(\xi'_j)$. By hypothesis assumption, 
    $\xi_i'\models t(\phi_2)$ implies $\xi_i\models t(\phi_2)$ and $\xi'_j\models \neg Tail \wedge t(\phi_1)$ implies 
    $\xi_j\models \neg Tail \wedge t(\phi_1)$. As a result, $\xi\models t(\phi)$ is true;
    
    \item If $\phi = \phi_1 R\phi_2$, then $t(\phi) = (Tail\vee t(\phi_1)) R t(\phi_2)$, and $\xi'\models t(\phi)$ implies that 
    for all $0\leq i<|\xi'|$ it holds that, $\xi'_i\models t(\phi_2)$ or there is $0\leq j<i$ such that 
    $\xi'_j\models Tail\vee t(\phi_1)$. Since $\xi = tp(\xi')$, so $\xi_i = tp(\xi'_i)$ for $0\leq i< |\xi|$. 
    By hypothesis assumption, $\xi'_i\models t(\phi_2)$ implies $\xi_i\models t(\phi_2)$ for 
    every $0\leq i<|\xi|-1$. Moreover, it is true that $Tail\in \xi[|\xi|-1]$, which implies $\xi[|\xi|-1]\models Tail\vee t(\phi_1)$. 
    Therefore, we have that $\xi\models t(\phi)$. 
\end{enumerate}
Because $tp(\xi')\models t(\phi)$ is true, and $tp(\xi')\models FTail$ is obviously true, we prove finally that 
$\xi=tp(\xi')\models \tnf{\phi}$.
\end{proof}

\begin{lemma}\label{lem:tnf2}
Let $\xi,\xi'$ are two non-empty finite traces satisfying $|\xi| = |\xi'|$ and $\xi'[i] = \xi[i]$ for $0\leq i< |\xi|-1$ as well as $\xi'[|\xi|-1] = \xi[|\xi|-1]\cup\{Tail\}$. Then $\xi\models \phi$ iff $\xi'\models \tnf{\phi}$.
\end{lemma}
\begin{proof}
  We prove by induction over the type of $\phi$. 
 
 \begin{enumerate}
    \item If $\phi$ is $\tt$, $\ff$ or a literal $l$, obviously $\xi\models\phi$ holds iff $\xi'\models\tnf{\phi}$ holds;
    
    \item If $\phi = \neg\psi$, then $\xi\models\phi$ holds iff $\xi\not\models\psi$ holds. By hypothesis assumption, $\xi\not\models\psi$ 
    holds iff $\xi'\not\models\tnf{\psi}$ holds, which means $\xi\models\phi$ holds iff $\xi'\models\tnf{\phi}$ holds;
    
    \item If $\phi = X\psi$, then $\xi\models\phi$ holds iff $|\xi|>1$ and $\xi_1\models\psi$ holds. By hypothesis assumption, 
    $\xi_1\models\psi$ holds iff $\xi'_1\models\tnf{\psi}$ holds, and $\xi'_1\models\tnf{\psi}$ holds iff 
    $\xi'\models\neg Tail\wedge X(\tnf{\psi})$ holds 
    (because $\neg Tail\in\xi'[0]$). As a result, we have the following equations:
    \begin{align*}
    & \xi\models X\psi\\
    & \Leftrightarrow\ \ \xi'\models\neg Tail\wedge X(\tnf{\psi})\\
    & \Leftrightarrow\ \ \xi'\models\neg Tail\wedge X(t(\psi)\wedge FTail)\\
    & \Leftrightarrow\ \ \xi'\models\neg Tail\wedge X(t(\psi))\wedge FTail
    \end{align*}
    Since $\tnf{\phi} = \neg Tail\wedge X(t(\psi))\wedge FTail$, so $\xi\models\phi$ iff $\xi'\models\tnf{\phi}$ is true;
    
    \item If $\phi = \phi_1\wedge\phi_2$, then $\xi\models\phi$ holds iff both $\xi\models\phi_1$ and $\xi\models\phi_2$ hold. 
    By hypothesis assumption, we have $\xi\models\phi_1$ holds iff $\xi'\models\tnf{\phi_1}$ holds, and $\xi\models\phi_2$ holds 
    iff $\xi'\models\tnf{\phi_2}$ holds. As a result, 
    $\xi\models\phi$ holds iff $\xi'\models\tnf{\phi_1}\wedge\tnf{\phi_2} = t(\phi_1)\wedge t(\phi_2)\wedge FTail = t(\phi_1\wedge \phi_2)\wedge FTail = \tnf{\phi_1\wedge\phi_2}$ holds;
    
    \item If $\phi =  \phi_1 U\phi_2$, then $\xi\models\phi$ holds iff there exists $0\leq i<|\xi|$ such that $\xi_i\models\phi_2$, and 
    for every $0\leq j< i$ it holds that $\xi_j\models\phi_1$. By hypothesis assumption, $\xi_i\models\phi_2$ holds iff 
    $\xi'_i\models\tnf{\phi_2}$ holds, and moreover, $\xi_j\models\phi_1$ holds iff $\xi'_j\models\tnf{\phi_1}$ holds. 
    Because of $0\leq j<i$ and $0\leq i<|\xi|$, $j$ does not equal to $|\xi|-1$, which means $\neg Tail\in\xi'[j]$. As a result, 
    $\xi'[j]\models \neg Tail\wedge \tnf{\phi_1}$. Therefore, $\xi'_i\models\phi_2$ holds and for every $0\leq j<i$, 
    $\xi'_j\models \neg Tail\wedge \tnf{\phi_1}$ is true, which means $\xi'\models (\neg Tail\wedge \tnf{\phi_1}) U \tnf{\phi_2}$ 
    is true. Finally, we have 
    \begin{align*}
    & \xi\models \phi_1 U\phi_2\\
    & \Leftrightarrow \xi'\models (\neg Tail\wedge \tnf{\phi_1}) U \tnf{\phi_2}\\
    & \Leftrightarrow \xi'\models (\neg Tail\wedge t(\phi_1)\wedge FTail) U (t(\phi_2)\wedge FTail)\\
    & \Leftrightarrow \xi'\models (\neg Tail\wedge t(\phi_1)) U t(\phi_2)\wedge FTail\\
    & \Leftrightarrow \xi'\models \tnf{\phi}
    \end{align*}
    The proof is done.
 \end{enumerate} 
\end{proof}

We are ready now to prove Theorem \ref{thm:tnf}. 

\begin{proof}
    ($\Rightarrow$) If $\phi$ is satisfiable, there is a non-empty finite trace $\xi$ such that $\xi\models\phi$. From Lemma 
    \ref{lem:tnf2}, we know that there is a corresponding finite trace $\xi'$ satisfying $|\xi| = |\xi'|$ and $\xi'[i] = \xi[i]$ for $0\leq i< |\xi|-1$ as well as $\xi'[|\xi|-1] = \xi[|\xi|-1]\cup\{Tail\}$ such that $\xi'\models\tnf{\phi}$. So $\tnf{\phi}$ 
    is satisfiable. 
    
    ($\Leftarrow$) If $\tnf{\phi}$ is satisfiable, there is a finite trace $\xi'$ satisfying $Tail\not\in \xi'[i]$ for $0\leq i< |\xi|-1$ and $Tail\in\xi'[|\xi|-1]$ such that $\xi'\models \tnf{\phi}$, from Lemma \ref{lem:tnf1}. 
    Moreover, according to Lemma \ref{lem:tnf2}, there is a corresponding 
    finite trace satisfying $|\xi| = |\xi'|$ and $\xi[i] = \xi'[i]$ for $0\leq i< |\xi|-1$ as well as $Tail\not\in \xi[|\xi|-1]$ 
    such that $\xi\models\phi$. So $\phi$ is satisfiable. 
\end{proof}

\subsection{Proof of Theorem \ref{thm:assign}}
\begin{proof}
$(\Rightarrow)$ Base case: when $\phi$ is a literal, Next, Unitl or Release formula, it is 
true since there is only one propositional assignment of $\phi^p$, i.e. $A= \{\phi\}$. 
Inductive step: if $\phi = \phi_1\wedge\phi_2$, $\xi\models\phi$ 
implies $\xi\models\phi_1$ and $\xi\models\phi_2$. By assumption hypothesis, there is $A_i$ of $\phi_i^p$ ($i=1,2$) such that 
$\xi\models\bigwedge A_i$. Let $A = A_1\cup A_2$, and a consistent $A$, in which either $\psi$ or $\neg \psi$ cannot be, 
must exists ($A$ may not be unique because $A_1$ and $A_2$ may not be unique). 
Otherwise, there is $\psi\in A_1$ and $\neg\psi\in A_2$ 
such that $\xi$ cannot model $\bigwedge A_1$ and $\bigwedge A_2$ at the same time, which is a contradiction. So $A$ is a propositional  assignment of $\phi^p$ and $\xi\models\bigwedge A$. The proof for $\phi=\phi_1\vee\phi_2$ is similar. 

$(\Leftarrow)$ $A$ is a propositional assignment of $\phi^p$, so $A\models\phi^p$ implies $(\bigwedge A)\Rightarrow \phi$. 
Therefore, $\xi\models \bigwedge A$ implies that $\xi\models \phi$.
\end{proof}

\subsection{Proof of Theorem \ref{thm:xnf}}
\begin{proof}
First, $\xnf{\phi}$ can be constructed recursively as follows:  (1) $\xnf{\phi} = \phi$, when $\phi$ is $\tt,\ff$, a literal or 
$\X\psi$ (Note $\phi$ is $\N$-free);
(2) $\xnf{\phi_1\ o\ \phi_2} = \xnf{\phi_1}\ o\ \xnf{\phi_2}$, where $o$ is $\wedge$ or $\vee$;
(3) $\xnf{\phi_1 \U\phi_2} = \xnf{\phi_2}\vee (\xnf{\phi_1}\wedge \X(\phi_1 \U\phi_2))$; and 
(4) $\xnf{\phi_1 \R\phi_2} = \xnf{\phi_2}\wedge (\xnf{\phi_1}\vee \X(\phi_1 \R\phi_2))$;
%
%
Since the construction is built on two expansion rules of Unitl and Release, and the expansion stops once the Until and Release are in the scope of Next, it preserves 
 the equivalence $\phi\equiv\xnf{\phi}$, and the cost is at most linear.
\end{proof}

\subsection{Proof of Lemma \ref{lem:reachable}}
\begin{proof}
Basically, for $s\in T(s_0, \sigma)$ ($\sigma\in \Sigma$), since there is a propositional assignment 
$A$ of $\xnf{\bigwedge s_0}^p$ such that $\sigma \supseteq L(A)$ and $s = X(A)$, $s$ is reachable from $s_0$ in one step. 
Inductively, assume $s$ is reachable from $s_0$ in $k$ ($k\geq 1$) steps. For $s'\in T(s, \sigma)$ ($\sigma\in \Sigma$), similarly 
we have $s'$ is reachable from $s$ in one step. As a result, $s'$ is reachable from $s_0$ in $k+1$ steps.
\end{proof}

\subsection{Proof of Theorem \ref{thm:reasoning}}
We first introduce the following lemma that is used for the proof.

\begin{lemma}\label{lem:final}
$s$ is a final state of $T_{\phi}$, iff there is a finite trace $\xi$ with $|\xi|=1$ such that $\xi\models s$. 
\end{lemma}
\begin{proof}
From Definition \ref{def:final}, $s$ is a final state iff there is a propositional assignment $A$ of the Boolean formula $Tail\wedge (\xnf{s})^p$ and 
$Tail\in A$. Recall that every Next subformula in $s$ is associated with $\neg Tail$, so $Tail\in A$ holds iff no Next subformula is in $A$, and thus iff $L(A)\models \xnf{s}^p$ holds.  Let $\xi = \sigma$ ($\sigma\in \Sigma$) such that $\sigma\supseteq L(A)$, 
and obviously $\xi\models s$. 
\end{proof}

Now we start to prove Theorem \ref{thm:reasoning}.
\begin{proof}
$(\Rightarrow)$ Since $\phi$ is satisfiable, there is a finite trace $\xi\models\phi$. Assume $|\xi| = n (n>0)$. Based on Theorem 
\ref{thm:assign}, there is a propositional assignment $A_0$ of $\xnf{\phi}^p$ such that $\xi\models \bigwedge A_0$. 
And according to Definition 
\ref{def:ts}, there is a transition $s_1\in T(s_0, \sigma_0)$ in $T_{\phi}$ where $s_0 = \phi$, $\sigma_0\supseteq L(A_0)$ and 
$s_1 = X(A_0)$. 
Moreover, we have that $\xi_1\models s_1$. Recursively,, we can prove that for $n>i\geq 0$, 
there is a transition $s_{i+1}\in T(s_i, \sigma_i)$ in $T_{\phi}$ such that $\sigma_i\supseteq L(A_i)$, $s_{i+1} = X(A_i)$ for some propositional assignment $A_i$ 
of $\xnf{s_i}^p$, and $\xi_{i+1}\models s_{i+1}$ holds. For $i=n-1$, since $|\xi_{i}| = 1$ and $\xi_{i}\models s_i$, $s_i$ is a final state according to Lemma \ref{lem:final}, and it is reachable from $s_0$ based on Lemma \ref{lem:reachable}.

    ($\Leftarrow$) Let $s$ be a final state in $T_{\phi}$, and it is reachable from the initial state $s_0$ from Lemma \ref{lem:reachable}. Assume a run $r=s_0,\ldots,s_{n-1}, s (n>=0)$ (when $n=0$, $s=s_0$ is the initial state) of $T_{\phi}$ on $\xi'=\sigma_0,\sigma_1,\ldots,\sigma_{n-1}$ leads from $\phi$ to $s$. Moreover according to Lemma \ref{lem:final}, there is a finite trace $\xi''$ with $|\xi''|=1$ such that $\xi''\models s$. Let $\xi = \xi'\cdot\xi'' = \sigma_0\sigma_1,\ldots\sigma_n (n\geq 0)$ where $\xi''=\sigma_n$, and now we prove that $\xi\models \phi$.  The proof can be achieved by induction from $n$ to $0$. Basically, $(\xi_n = \sigma_n)\models s$ is obviously true. Inductively assume $\xi_i\models s_{i}$ for $n\geq i\geq 1$, so $\xi_{i-1}= \xi[i-1]\cdot\xi_i$ satisfies $\xi[i-1]\supseteq L$ and $\xi_i\models s_i$ for some $s_{i}\in T(s_{i-1}, L)$ from the definition of $T_{\phi}$, which means $\xi_{i-1}\models s_{i-1}$. When $i=0$, we prove that $(\xi = \xi_0)\models (s_0=\phi)$.
\end{proof}

\subsection{Proof of Lemma \ref{lem:cdlsc}}
\begin{proof}
First, \cdlsc sets $\cs[0] = \{\phi\}$ after checking $Tail\wedge\xnf{\phi}^p$ is unsatisfiable, which meets Item 2 of Definition \ref{def:cs}. Secondly after each iteration $i\geq 0$, $try\_satisfy$ guarantees that $\{\phi\}$ is added into each $\cs[i]$ if no model is found, which meets Item 1 of Definition \ref{def:cs}. 
By enumerating Line \ref {alg:satisfy:uc0end} and \ref{alg:satisfy:ucend} in $try\_satisfy$, we have that $\xnf{s}\wedge \neg \X(\cs[i])$  is unsatisfiable for $s\in\cs[i+1](0\leq i\leq |\cs|-1)$, which meets Item 3 of Definition \ref{def:cs}.  So $\cs$ is a conflict sequence after each iteration with no model found.
\end{proof}

\subsection{Proof of Theorem \ref{thm:terminate}}
\begin{proof}
    \cdlsc runs iteratively, so \cdlsc terminates iff either the procedure $try\_satisfy$ or $inv\_found$  returns true for some iteration. From Lemma \ref{lem:cdlsc}, $\cs$ is a conflict sequence after each iteration if no model found .          
    After each iteration, $try\_satisfy$ returns true iff a final state 
    is found (Line \ref{alg:satisfy:satstart}-\ref{alg:satisfy:satend}) based on Theorem \ref{thm:sat}. 
    Meanwhile, $inv\_found$ returns true iff $\phi$ is 
    unsatisfiable because of Theorem \ref{thm:unsat}. As a result, there is always such an iteration, after which \cdlsc can 
    terminate and terminate correctly.
\end{proof}